\documentclass[journal]{IEEEtran}

\ifCLASSINFOpdf

\else

\fi

\makeindex             

\usepackage{amsthm}
\usepackage{epsfig} 
\usepackage{epstopdf}
\usepackage{amsmath}
\usepackage{amsthm}
\usepackage{comment}
\usepackage{url}
\usepackage{algorithm}
\usepackage{algorithmic}

\usepackage{amsfonts}
\usepackage{bm}
\usepackage{color}
\usepackage{cite}
\usepackage{stfloats}
\usepackage{chemarrow}
\usepackage{extarrows}

\usepackage{url}

\usepackage{setspace}

\allowdisplaybreaks[4]

\DeclareGraphicsExtensions{.pdf,.png,.jpg,.eps,.ps}

\def\baa{\begin{align}}
\def\eaa{\end{align}}

\newcommand{\bsq}{\begin{subequations}}
\newcommand{\esq}{\end{subequations}}

\newcommand{\beq}{\begin{equation}}
\newcommand{\eeq}{\end{equation}}
\newcommand{\bq}{\begin{eqnarray}}
\newcommand{\eq}{\end{eqnarray}}
\newcommand{\bqn}{\begin{eqnarray*}}
	\newcommand{\eqn}{\end{eqnarray*}}
\newcommand{\bee}{\begin{enumerate}}
	\newcommand{\eee}{\end{enumerate}}
\newcommand{\bi}{\begin{itemize}}
	\newcommand{\ei}{\end{itemize}}

\newboolean{showcomments}
\setboolean{showcomments}{true}
\newcommand{\wang}[1]{\ifthenelse{\boolean{showcomments}}
	{ \textcolor[rgb]{1,0,1}{(ZW:  #1)}}{}}
\newcommand{\fliu}[1]{\ifthenelse{\boolean{showcomments}}
	{ \textcolor{red}{(FL:  #1)}}{}}
\newcommand{\zhao}[1]{\ifthenelse{\boolean{showcomments}}
	{ \textcolor{green}{(JP:  #1)}}{}}
\newcommand{\slow}[1]{\ifthenelse{\boolean{showcomments}}
	{ \textcolor{blue}{(SL:  #1)}}{}}

\theoremstyle{definition}
\newtheorem{theorem}{Theorem}
\newtheorem{lemma}[theorem]{Lemma}

\theoremstyle{definition}
\newtheorem{definition}{Definition}
\newtheorem{remark}{Remark}

\newtheorem{assumption}{\textit{Assumption}}

\ifCLASSOPTIONcaptionsoff
\usepackage[nomarkers]{endfloat}
\let\MYoriglatexcaption\caption
\renewcommand{\caption}[2][\relax]{\MYoriglatexcaption[#2]{#2}}
\fi

\allowdisplaybreaks

\begin{document}
\setstretch{0.99}	
	\title{Asynchronous Distributed Power Control of Multi-Microgrid Systems Based on \\the Operator Splitting Approach}
	
	\author{Zhaojian~Wang, 
		Shengwei~Mei,~\IEEEmembership{Fellow,~IEEE}, 
		Feng~Liu,~\IEEEmembership{Senior Member,~IEEE},\\
		Peng Yi, Ming Cao,~\IEEEmembership{Senior Member,~IEEE} 
	\thanks{This work was supported  by the National Natural Science Foundation
		of China ( No. 51677100, U1766206, No. 51621065 ). (\textit{Corresponding author: Feng Liu})    }       	
	\thanks{Z. Wang, F. Liu and S. Mei  are with the State Key Laboratory of Power System and Department
		of Electrical Engineering, Tsinghua University, Beijing, 100084,
		China (e-mail: lfeng@tsinghua.edu.cn).}
	\thanks{P. Yi is with the Department of Electrical and Systems Engineering, Washington University in St. Louis, 1 Brookings Drive, St. Louis, MO, 63130, USA.}
	\thanks{M. Cao is with the Faculty of Science and Engineering, University of Groningen, Groningen 9747 AG, The Netherlands.}
}


	\maketitle
	
	\begin{abstract}                                                
		Forming (hybrid) AC/DC microgrids (MGs) has become a promising manner for  the interconnection of various kinds of distributed generators that are inherently AC or DC electric sources. This paper addresses the distributed asynchronous power control problem of hybrid microgrids, considering imperfect communication due to non-identical sampling rates and communication delays. To this end,  we first formulate the optimal power control  problem of MGs and devise a synchronous algorithm. Then, we analyze the impact of asynchrony on optimal power control and propose an asynchronous iteration algorithm based on the synchronous version. By introducing a random clock at each iteration, different types of asynchrony are fitted into a unified framework, where the asynchronous algorithm is converted into a fixed-point  problem based on the operator splitting method, leading to a convergence proof. We further provide an upper bound estimation of the time delay in the communication. Moreover, the real-time implementation of the proposed algorithm in both  AC and DC MGs is introduced. By taking the power system as a solver, the controller is simplified by reducing one order and the power loss can be considered. Finally, a benchmark MG is utilized to verify the effectiveness and advantages of the proposed algorithm.
	\end{abstract}
	
	
	{\setlength{\parskip}{0.5\baselineskip}
	\begin{IEEEkeywords}
		Asynchronous control, distributed power control; hybrid AC/DC microgrids, time delay.
	\end{IEEEkeywords}}

	\section{Introduction}
	Multi-Microgrid systems or Microgrids (MGs) are clusters of distributed generators (DGs), energy storage systems and loads, which are generally categorized into three types: AC, DC and hybrid AC/DC MGs \cite{xu2017decentralized,xu2018prescribed,weng2018distributed}. A hybrid AC/DC MG has the great advantage of reducing processes of multiple inverse conversions in the involved individual AC or DC grid \cite{xia2018power}. In this paper, we address the distributed power control problem of hybrid AC/DC MGs considering asynchrony.
	
	Traditionally, a hierarchical control structure is utilized in MGs for power control, which is composed of primary control, secondary control and tertiary control \cite{guerrero2011hierarchical}, usually in a centralized way. Such a centralized control architecture, however, may face great challenges raised by ever-increasing uncertain and volatile renewable generations that require  fast response of controllers \cite{wang2018distributed}. On the other hand, as MGs usually belong to different owners, privacy concerns may prevent the control center acquiring  information from individual MGs. In this context, breaking the hierarchy of MG control architecture becomes an emerging research topic, supported by the new idea that real-time coordination could be embedded in the local steady-state optimization of individual MGs by exchanging information only between neighboring MGs. This essentially advocates a \emph{distributed} control paradigm \cite{dorfler2016breaking,nguyen2017distributed,li2016event,li2017distributed}. 
		
	Different distributed strategies have been developed in literatures for optimal real-time coordination, which can roughly be divided	into two categories in terms of methodology: consensus based methods \cite{ahn2018distributed, Simpson2015Secondary, Dehkordi2018distributed, zhou2017event,Chen2015Distributed, wu2017distributed, zhao2018analysis}, and (sub)gradient based decomposition methods \cite{zhang2015real,wu2017distributed2,   wang2017unified,zholbaryssov2019resilient}. In the consensus based control, the agents manage to estimate the global variable using a consensus algorithm \cite{meng2014studies,li2017distributed}. Specifically, in power systems, the global variable could be the generation ratio and the marginal cost. The former implies that all generators have the same generation ratio with respect to its maximal capability \cite{Simpson2015Secondary, Dehkordi2018distributed, zhou2017event}. The latter implies that all generators share the same marginal cost, and hence the generation configuration is economically optimal \cite{Chen2015Distributed, wu2017distributed, zhao2018analysis}. Even though the consensus methods are easy to be implemented, they are difficult to address complicated (global) constraints. In this situation, (sub)gradient based decomposition methods could be applied, where the  optimization problem is solved by dual (sub)gradient ascent \cite{zhang2015real,yuan2016regularized,wu2017distributed2,   wang2017unified,zholbaryssov2019resilient}.
	
	Although the aforementioned methods have achieved great success, there are still two issues to be hurdled before a practical implementation. First, most of them have considered only {synchronous} distributed control. Therefore,  all MGs must carry out computation simultaneously, implying that a global clock is necessary to ensure the instants for control actions getting strictly synchronized. This is  computationally inefficient and impractical since in each iteration all MGs have to wait for the slowest one to finish before executing their local actions in the next iteration. In fact, asynchrony widely exists in power systems, such as time delays and non-identical sampling rates \cite{alkano2017asynchronous,yang2017distributed}.	
	Hence, the synchrony requirement limits the application of distributed control.
	Second, the load demand in existing literature is usually assumed to be known, especially in those papers addressing optimal economic dispatch problem \cite{wu2017distributed2,zhao2018analysis,zholbaryssov2019resilient}. However, the load demand is very difficult to measure and is always time varying when demand response and electric vehicles exist. In addition, the fast varying situation requires the real-time implementation of the algorithm. This somewhat makes the existing method difficult to be applied. Thus, it is significant to investigate the controller that adapts to {asynchronous} and real-time implementation.
	
	The purpose of this paper is to propose an asynchronous algorithm for the optimal power control of hybrid AC/DC MGs. In addition, we will also introduce the real-time implementation of the algorithm, where the power system is taken as a solver to compute some variables automatically. In this way, the algorithm can be greatly simplified and the power loss can be considered. Main contributions of this paper are as follows. 
	\begin{itemize}
		\item We devise an asynchronous distributed algorithm  to solve the optimal power control problem of hybrid AC/DC MGs. Different from most existing works, in this paper, different kinds of asynchrony are fitted into one unified form  i.e., the time interval between two consecutive iterations,  by introducing a virtual  random clock; 
		\item We give a  convergence proof of the distributed algorithm by  converting it into a fixed-point iteration using the operator splitting approach. The upper bound of communication delay that guarantees the convergence is given, which is approximately proportional to the square root of the number of MGs;
		\item We provide a real-time implementation of the proposed algorithm, where the power system is taken as a solver to compute some variables automatically. This simplifies the algorithm by reducing one oder. Moreover, we provide methods to estimate the unknown load demand in AC and DC MGs with physical system variables respectively.  In this way, the impact of power loss such as line and inverter loss can  be considered.
	\end{itemize}

The rest of this paper is organized as follows. In Section II, some notations and preliminaries are introduced. Section III formulates the power dispatch problem in hybrid MGs and proposes the synchronous algorithm. In Section IV, different types of asynchrony are introduced and an asynchronous algorithm is proposed. The optimality of its equilibrium point and convergence of the asynchronous algorithm are proved in Section V. The implementation method in hybrid MGs is introduced in Section VI. We confirm the performance of the controller via simulations on a benchmark low voltage MG system in Section VII. Section VIII concludes the paper.

\section{Notations and Preliminaries}

\textit{Notations}: A hybrid MG system is composed of a cluster of AC and DC MGs connected by lines. Each MG is treated as a bus with both generation and load. Denote AC MGs by $\mathcal {N}_{ac}=\{1, 2, \dots, n_{ac}\}$, and DC MGs by $\mathcal {N}_{dc}=\{n_{ac}+1, n_{ac}+2, \dots, n_{ac}+n_{dc}\}$. Then the set of MG buses is $\mathcal{N}=\mathcal {N}_{ac} \cup \mathcal {N}_{dc}$.
Let $\mathcal E\subseteq \mathcal N\times \mathcal N$ be the set of lines, where $(i,k)\in \mathcal E$ if MGs $i$ and $k$ are connected directly. Then the overall system is modeled as a connected graph $\mathcal{G}:=(\mathcal N, \mathcal E)$. Besides the physical connection among MGs, we also define a communication graph for MGs. Denote by $N_i$ the set of informational neighbors of MG $i$ over the communication graph, implying MGs $i, j$ can communicate if and only if $j\in N_i$. Denote by $N_i^2$ the set of two-hop neighbors of MG $i$ over the communication graph. The cardinality of $N_i$ is denoted by $|N_i|$. The communication graph is also assumed to be undirected and connected, which could be different from the physical graph. Denote by $L$ the Laplacian matrix of communication graph. 

\textit{Preliminaries}: In this paper, $\mathbb{R}^n$ ($\mathbb{R}^n_{+}$) is the $n$-dimensional (nonnegative) Euclidean space. For a column vector $x\in \mathbb{R}^n$ (matrix $A\in \mathbb{R}^{m\times n}$), $x^{\mathrm{T}}$($A^{\mathrm{T}}$) denotes its transpose. For vectors $x,y\in \mathbb{R}^n$, $x^{\mathrm{T}}y=\left\langle x,y \right\rangle$ denotes the inner product of $x,y$. 
$\left\|x \right\|=\sqrt{x^{\mathrm{T}}x}$ denotes the Euclidean norm of $x$. 
For a positive definite matrix $G$, denote the inner product $\left\langle x,y \right\rangle_G=\left\langle Gx,y \right\rangle$. Similarly, the $G$-matrix induced norm $\left\|x \right\|_G=\sqrt{\left\langle Gx,x \right\rangle}$. 
Use $I$ to denote the identity matrix with proper dimensions. For a matrix $A=[a_{ij}]$, $a_{ij}$ stands for the entry in the $i$-th row and $j$-th column of $A$. Use $\prod_{i=1}^n\Omega_i$ to denote the Cartesian product of the sets $\Omega_i, i=1, \cdots, n$. 
Given a collection of $y_i$ for $i$ in a certain set $Y$, $y$ denotes the column vector
$y := (y_i, i\in Y)$ with a proper dimension with $y_i$ as its components.

Define the projection of $x$ onto a set $\Omega$ as 
\begin{align}
\label{def:projection}
	\mathcal{P}_{\Omega}(x)=\arg \min_{y\in \Omega}\left\|x-y \right\|
\end{align}
Use ${\rm{Id}}$ to denote the identity operator, i.e., ${\rm{Id}}(x)=x$, $\forall x$. Define $N_\Omega(x)=\{v|\left\langle v, y-x\right\rangle\le 0, \forall y\in \Omega\}$. We have $\mathcal{P}_{\Omega}(x)=({\rm{Id}}+N_{\Omega})^{-1}(x)$ \cite{yi2017distributed}, \cite[Chapter 23.1]{bauschke2011convex}.

%


For a single-valued operator $\mathcal{T}:\Omega\subset\mathbb{R}^n\rightarrow\mathbb{R}^n$, a point $x\in \Omega$ is a fixed point of $\mathcal{T}$ if $\mathcal{T}(x)\equiv x$. The set of fixed points of $\mathcal{T}$ is denoted by ${Fix}(\mathcal{T})$. $\mathcal{T}$ is nonexpansive if $\left\|\mathcal{T}(x)- \mathcal{T}(y)\right\|\le\left\|x- y\right\|, \forall x, y \in \Omega$. For $\alpha\in (0,1)$, $\mathcal{T}$ is called $\alpha$-averaged if there exists a nonexpansive operator $\mathcal{R}$ such that $\mathcal{T}=(1-\alpha){\rm{Id}}+\alpha \mathcal{R}$. We use $\mathcal{A}(\alpha)$ to denote the class of $\alpha$-averaged operators. For $\beta\in\mathbb{R}^1_{+}$, $\mathcal{T}$ is called $\beta$-cocoercive if $\beta\mathcal{T}\in \mathcal{A}(\frac{1}{2})$.

\section{Synchronous Distributed Algorithm}

In this section, we introduce the economic dispatch problem in MGs and propose a synchronous algorithm.
\subsection{Economic dispatch model}
The power dispatch is to achieve the power balance in MGs while minimizing the generation cost, which can be formulated as the following optimization problem
\begin{subequations}
	\label{eq:opt.general1}     
	\begin{align}
	\min\limits_{P_i^g}\quad&  \sum\limits_{i\in \mathcal{N}} f_i(P_i^g)
	\label{eq:opt.general1d}
	\\ 
	\text{s.t.} \quad
	&\sum\limits_{i\in \mathcal{N}} P_i^g = \sum\limits_{i\in \mathcal{N}} P_i^d
	\label{eq:opt.general1a}
	\\
	\label{eq:opt.general1b}
	&\underline P_i^g\le P_i^g\le \overline P_i^g
	\end{align}
\end{subequations}
where $f_i(P_i^g)=\frac{1}{2}a_i(P_i^g)^2 + b_iP_i^g$, with $a_i>0, b_i>0$. $P_i^g$, $P_i^d$ are the power generation and load demand of MG $i$ respectively. $\underline P_i^g, \overline P_i^g$ are the lower and upper bounds of $P_i^g$ respectively. The objective function \eqref{eq:opt.general1d} is to minimize the total generation cost of the MGs. Constraint \eqref{eq:opt.general1a} is the power balance over MGs. And \eqref{eq:opt.general1b} is the generation limit of each MG. 

For the optimization problem \eqref{eq:opt.general1}, we make the following assumption.
\begin{assumption}
	\label{Slater}
	The Slater's condition  \cite[Chapter 5.2.3]{boyd2004convex} of \eqref{eq:opt.general1} holds, i.e., problem \eqref{eq:opt.general1} is feasible due to affine constraints.
\end{assumption}

\subsection{Synchronous Algorithm}
\label{subsec:optimality.1}
The Lagrangian function of \eqref{eq:opt.general1} is 
\begin{align*}
\mathcal{L}&= \sum\limits_{i\in \mathcal{N}} f_i(P_i^g) +\mu\left(\sum\limits_{i\in \mathcal{N}} P_i^g - \sum\limits_{i\in \mathcal{N}} P_i^d\right) \\
&\quad + \sum\limits_{i\in \mathcal{N}}\gamma_i^{-} (\underline P_i^g- P_i^g) + \sum\limits_{i\in \mathcal{N}}\gamma_i^{+} (P_i^g- \overline P_i^g)
\end{align*}
where $\mu, \gamma_i^{-}, \gamma_i^{+}$ are Lagrangian multipliers. Here $\mu$ is a global variable, but will be estimated by individual MGs locally.

Define the sets 
\begin{align}
\Omega_i:=\left\{P_i^g\ |\ \underline P_i^g\le P_i^g\le \overline P_i^g\right\},\ \Omega=\prod\nolimits_{i=1}^N\Omega_i
\end{align}

Then, we give the \underline{s}ynchronous \underline{d}istributed algorithm for \underline{p}ower \underline{d}ispatch (SDPD). 
In this case, the update of agent $i$ at iteration $k$ is given as, which takes the form of Krasnosel'ski{\v{i}}-Mann iteration \cite[Chapter 5.2]{bauschke2011convex}.
\begin{subequations}
	\label{SYN}
	\begin{align}
	\label{SYN1}
	&\tilde \mu_{i,k}= \mu_{i,k}+\sigma_{\mu}\bigg( -\sum\limits_{j\in N_i}\left(\mu_{i,k}-\mu_{j,k}\right) \nonumber\\
	& \qquad\qquad \qquad+ \sum\limits_{j\in N_i}(z_{i,k}-z_{j,k}) +P^g_{i,k}- P^d_i \bigg) \\
	\label{SYN2}
	&\tilde z_{i_k,k}= z_{i_k,k}-\sigma_{z} \bigg(2\sum\limits_{j\in N_i}\left(\tilde\mu_{i,k}-\tilde\mu_{j,k}\right) \nonumber \\
	&\qquad\qquad\qquad\qquad\qquad  -\sum\limits_{j\in N_i}\left(\mu_{i,k}-\mu_{j,k}\right) \bigg)\\
	\label{SYN3}
	&\tilde P^g_{i,k}=\mathcal{P}_{\Omega_i}\left(P^g_{i,k} - \sigma_{g}\big(f_i^{'}(P^g_{i,k})+2\tilde \mu_{i,k} -\mu_{i,k}\big)\right) \\
	\label{SYN4}
	& \mu_{i,k+1}= \mu_{i,k} + \eta_k \left(\tilde \mu_{i,k}-\mu_{i,k}\right)\\
	\label{SYN5}
	& z_{i,k+1}= z_{i,k} + \eta_k \left(\tilde z_{i,k}-z_{i,k}\right)\\
	\label{SYN6}
	& P^g_{i,k+1}= P^g_{i,k}+ \eta_k \left(\tilde P^g_{i,k}-P^g_{i,k}\right)
	\end{align}
\end{subequations}
where $\sigma_{\mu}, \sigma_{z}, \sigma_{g}, \eta_k$ are positive constants, and $\sigma_{\mu}, \sigma_{z}, \sigma_{g}$ are supposed to be chosen such that $\Phi$ in \eqref{SYNc0} (given in Section V.B) is positive definite.

%
%
%
%
%
%
%
%
%
%
%

In \eqref{SYN}, the load demand $P_i^d$ is usually difficult to know. We will provide a practical method to estimate $P^g_i-{P}^d_i$ instead of directly measuring $P_i^d$  in the implementation, as explained in Section VI. Later in Section IV, we will show that the SDPD is simply a special case of the asynchronous algorithm. Therefore, its properties, such as the optimality of the equilibrium point and the convergence, are immediate consequence of the results of asynchronous algorithm, which are skipped here.

\section{Distributed Asynchronous Algorithm}
In this section, we first introduce several typical types of asynchrony existing in MGs. Then, we devise an asynchronous algorithm by modifying Algorithm SDPD.
\subsection{Asynchrony in Microgrids}
In SDPD, each MG gathers information, computes locally and conveys new information to its neighbors over the communication graph. In this process, asynchrony may arise in each step. When gathering information, individual MGs may have different sampling rates, which results in non-identical computation rates accordingly. In addition, other imperfect communication situations such as time delay caused by  congestion or even failure are very common in power systems, which essentially result in asynchrony. 

In  synchronous computation, an MG has to wait for the slowest neighbor to complete the computation by inserting certain idle time. Communication delay, congestion or even package loss can further lengthen the waiting time. This process is illustrated in Fig.\ref{fig:async}(a). Thus, the slowest MG and communication channel may cripple the system in the synchronous execution. In contrast, the MGs with asynchronous computation do not need to wait and computes continuously with little idling, as shown in Fig.\ref{fig:async}(b). Even if some of its neighbors fail to update in time, the MG can use the previously stored information. That means, the MG could  execute an iteration without the latest information from its neighbors. 

\begin{figure}[t]
	\centering
	\includegraphics[width=0.45\textwidth]{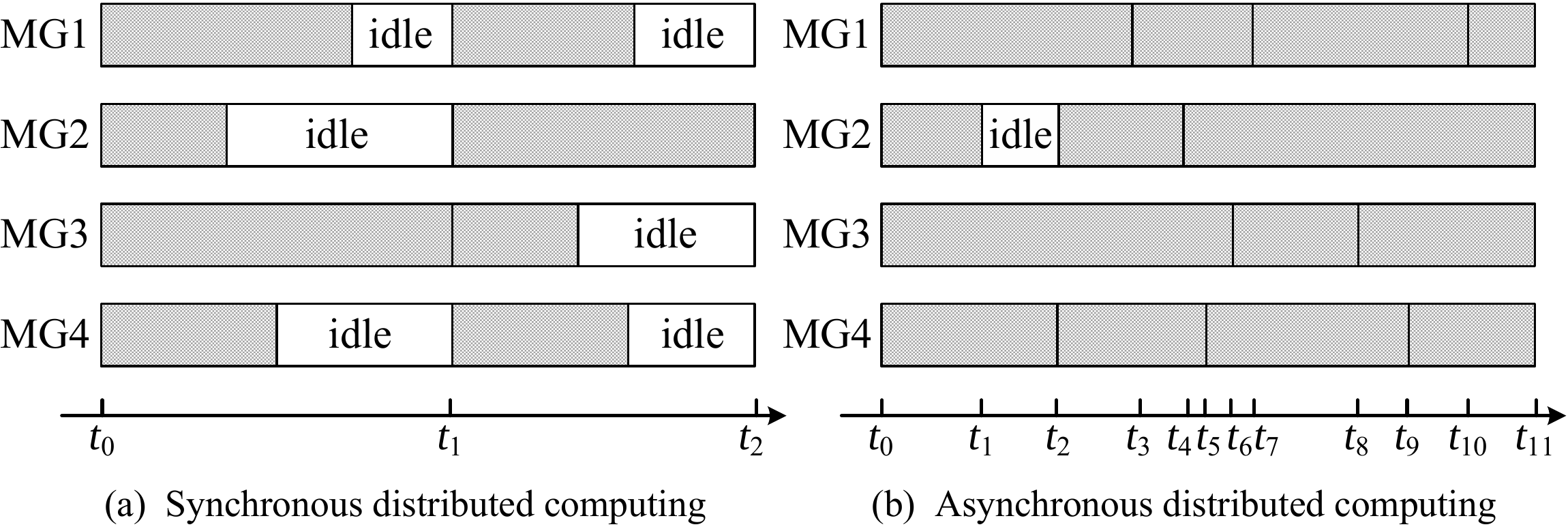}
	\caption{Synchronous versus asynchronous computation}
	\label{fig:async}
\end{figure}

\subsection{Asynchronous Algorithm}
In this subsection, we propose an \underline{as}ynchronous \underline{d}istributed algorithm for \underline{p}ower \underline{d}ispatch (ASDPD) based on SDPD. Different from the iteration number $k$ in \eqref{SYN}, here each MG has its own iteration number $k_i$, implying that a \emph{local clock} is used instead of the global clock. At each iteration $k_i$, MG $i$ computes in the following way.
\begin{subequations}
	\label{ASYN}
\begin{align}
		\label{ASYN1}
		&\tilde \mu_{i,k_i}= \mu_{i,k_i-\tau_{i}^{k_i}}+\sigma_{\mu}\bigg( - \sum\limits_{j\in N_i}\big(\mu_{i,k_i-\tau_{i}^{k_i}}-\mu_{j,k_j-\tau_{j}^{k_j}}\big) \nonumber \\
		&\quad + \sum\limits_{j\in N_i}\big(z_{i,k_i-\tau_{i}^{k_i}}-z_{j,k_j-\tau_{j}^{k_j}}\big) +P^g_{i,k_i-\tau_{i}^{k_i}}- P^d_i \bigg)
		\\
		\label{ASYN2}
		&\tilde z_{i,k_i}= z_{i,k_i-\tau_{i}^{k_i}} -\sigma_{z} \bigg(\sum\limits_{j\in N_i } \big(\mu_{i,k_i-\tau_{i}^{k_i}}-\mu_{j,k_j-\tau_{j}^{k_j}}\big)  \nonumber \\
		& \qquad - \sum\limits_{j\in N_i\cup N_i^2}2\sigma_{\mu}\ell_{ij}\big(\mu_{i,k_i-\tau_{i}^{k_i}}-\mu_{j,k_i-\tau_{j}^{k_j}}\big)\nonumber\\
		& \qquad + \sum\limits_{j\in N_i\cup N_i^2}2\sigma_{\mu}\ell_{ij}\big(z_{i,k_i-\tau_{i}^{k_i}}-z_{j,k_j-\tau_{j}^{k_j}}\big)\nonumber\\
		& \qquad + \sum\limits_{j\in N_i}2\sigma_{\mu}\big(P^g_{i,k_i-\tau_{i}^{k_i}}- P^d_i\big) \bigg) 
		\\
		\label{ASYN3}
		\begin{split}
		&\tilde P^g_{i,k_i}=\mathcal{P}_{\Omega_i}\left(P^g_{i,k_i-\tau_{i}^{k_i}} - \sigma_{g}\big(f_i^{'}(P^g_{i,k_i-\tau_{i}^{k_i}})+2\tilde \mu_{i,k_i} \right. \\ &\left.\qquad\qquad\qquad\qquad\qquad\qquad\quad \qquad-\mu_{i,k_i-\tau_{i}^{k_i}}\big)\right) \end{split}		\\
		\label{ASYN4}
		& \mu_{i,k_i+1}= \mu_{i,k_i-\tau_{i}^{k_i}} + \eta_k \big(\tilde \mu_{i,k_i}-\mu_{i,k_i-\tau_{i}^{k_i}}\big)\\
		\label{ASYN5}
		& z_{i,k_i+1}= z_{i,k_i-\tau_{i}^{k_i}} + \eta_k \big(\tilde z_{i,k_i}-z_{i,k_i-\tau_{i}^{k_i}}\big)\\
		\label{ASYN6}
		& P^g_{i,k_i+1}= P^g_{i,k_i-\tau_{i}^{k_i}}+ \eta_k \big(\tilde P^g_{i,k_i}-P^g_{i,k_i-\tau_{i}^{k_i}}\big)
\end{align}
\end{subequations}
where $ \ell_{ij} $ is the $i$th row and $j$th column element of matrix $L^2=L\times L$ and $\tau_{i}^k$ is the random time delay. $ \ell_{ij} \neq 0 $ holds only if $ j\in N_i\cup N_i^2   $ \cite{anderson2018distributed}. Denote $w=(\mu, z, P^g)$.  $w_{i,k_i-\tau_{i}^{k_i}}$ is the state of MG $i$ at iteration $k$, and $w_{j,k_j-\tau_{j}^{k_j}}$ is the latest information obtained from MG $j$.
Considering each MG has its own local clock, we  have the following asynchronous algorithm. 
\begin{algorithm}
	\caption{\textit{ASDPD}}
	\label{algorithm1}
	\textbf{Input:} For MG $i$, the input is $\mu_{i,0}, z_{i,0}\in\mathbb{R}^n$, $P_{i,0}^g\in \Omega_i$.
	
	
	\textbf{Iteration at} \textit{$k_i$}: Suppose MG $i$'s clock ticks at time $k_i$, then MG $i$ is activated and updates its local variables as follows: 
	
	$\ \ $\textbf{Step 1: }\textbf{Reading phase}
	
	$\quad$Get $\mu_{j,k_j-\tau_{j}^{k_j}}, z_{j,k_j-\tau_{j}^{k_j}}$ from its neighbors' and two-hop neighbors' output cache. 
	
	$\ \ $\textbf{Step 2:} \textbf{Computing phase}
	
	$\quad$Calculate $\tilde \mu_{i,k_i}, \tilde z_{i,k_i}$ and $\tilde P^g_{i,k_i}$ according to \eqref{ASYN1}, \eqref{ASYN2} and \eqref{ASYN3} respectively.
	
	$\quad$Update $ \mu_{i,k_i+1}, z_{i,k_i+1}$ and $ P^g_{i,k_i+1}$ according to \eqref{ASYN4}, \eqref{ASYN5} and \eqref{ASYN6} respectively.
	
	$\ \ $\textbf{Step 3:} \textbf{Writing phase}
	
	$\quad$Write $\mu_{i,k_i+1}$, $z_{i,k_i+1}$ to its output cache and  $\mu_{i,k_i+1}$, $z_{i,k_i+1}$, $P^g_{i,k_i+1}$ to its local storage. Increase $k_i$ to $k_i+1$.
\end{algorithm}

\begin{remark}
	In Algorithm \ref{algorithm1}, if MG $i$ is activated, it will read the latest information from its neighbors.  Even if some neighbors are not accessible in time due to communication issue, it can still execute the iteration by using the previous information stored in its input cache. Despite asynchrony caused by different reasons, MG $i$ only concerns whether the latest information comes, which implies that their effect can be characterized by the time interval between two successive iterations. Thus, our algorithm can admit different types of asynchrony. 	
\end{remark}
\begin{remark}
	As the element $ \ell_{ij} \neq 0 $ holds only if $ j\in N_i\cup N_i^2   $, the ASDPD is still distributed. Similar settings are also used in \cite{tang2019fast,ramirez2018distributed,anderson2018distributed}.
	However, it may make the communication graph denser. In the Section VI, we will show that the power system can be treated as a part of solver. Then, we can carry out the ASDPD by local measurement and neighboring communication. 
\end{remark}

\section{Optimality and Convergence Analysis}
In this section, we analyze the optimality of the equilibrium point of dynamic system \eqref{ASYN}, as well as the convergence of Algorithm \ref{algorithm1}. To this end, we need to introduce a sequence of global iteration numbers that serve as a reference \emph{global clock} to unify the local iterations of individual MGs in a coherence manner \cite{cao2008agreeing}. Note that the global clock is only used for convergence analysis, but not required in ASDPD. 

Specifically, we queue $k_i$ of all MGs in the order of  time, and use a new number $k$ to denote the $k$th iteration in the queue. This treatment is shown in Fig.\ref{merge_k} by taking two MGs as an example. Suppose that, at the iteration $k$, the probability that MG $i$ is activated to update its local variables follows a uniform distribution. Hence, each MG is activated with the same probability, which simplifies the convergence proof.

\begin{figure}[htp]
	\centering
	\includegraphics[width=0.42\textwidth]{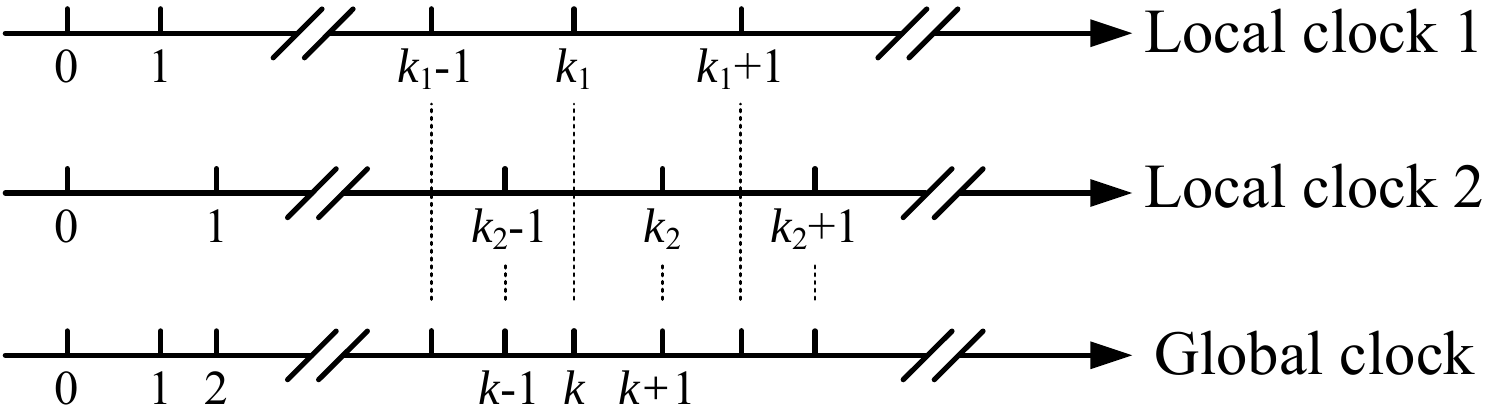}
	\caption{Local clocks versus global clock}
	\label{merge_k}
\end{figure}

To prove the convergence, we first convert the synchronous algorithm to a fixed-point iteration with an averaged operator. Then a \emph{nonexpansive} operator is constructed, leading to the convergence results of the asynchronous algorithm. Finally, we provide an upper bound of the time delay.

\subsection{Algorithm Reformulation}
If the time delay is not considered, \eqref{ASYN} is degenerated to \eqref{SYN}. 
In this sense, the SDPD is a special case of ASDPD, and we only need to analyze the property of ASDPD. 
The compact form of \eqref{ASYN1} - \eqref{ASYN6} without delay, i.e., \eqref{SYN1}-\eqref{SYN6}, is
\begin{subequations}
	\begin{align}
		\label{SYNc1}
		\begin{split}
		&\tilde \mu_{k}= \mu_{k}+\sigma_{\mu}\left( - L \cdot \mu_{k} +L \cdot z_{k} +P^g_{k}- P^d \right) \end{split}\\
		\label{SYNc2}
		&\tilde z_{k}= z_{k}+\sigma_{z}\left( - 2 L \cdot {\tilde \mu}_{k} +L \cdot \mu_{k} \right)\\
		\label{SYNc3}
		&\tilde P^g_{k}=\mathcal{P}_{\Omega}\left(P^g_{k} - \sigma_{g}\left(\nabla f(P^g_{k})+2\tilde \mu_{k} -\mu_{k}\right)\right) \\
		\label{SYNc4}
		& \mu_{k+1}= \mu_{k} + \eta_k \left(\tilde \mu_{k}-\mu_{k}\right)\\
		\label{SYNc5}
		& z_{k+1}= z_{k} + \eta_k \left(\tilde z_{k}-z_{k}\right)\\
		\label{SYNc6}
		& P^g_{k+1}= P^g_{k}+ \eta_k \left(\tilde P^g_{k}-P^g_{k}\right)
	\end{align}
\end{subequations}
where, $\nabla f(P^g_{k})$ is the gradient of $f(P^g_{k})$. The subscript $ k_i $ is substitute by a global notation $k$. The equation \eqref{SYNc2} is obtained by combing \eqref{ASYN1} with \eqref{ASYN2}, and others are straightforward.

Next we show that $\eqref{SYNc1}-\eqref{SYNc6}$ can be converted into a fixed-point iteration problem with an averaged operator \cite{yi2018asynchronous,yi2017distributed}.

Equation \eqref{SYNc1} is equivalent to 
\begin{align}
\label{SYNd1}
	- L \cdot \mu_{k} - P^d&=-P^g_{k}-L\cdot z_{k}+\sigma_{\mu}^{-1}(\tilde \mu_{k}- \mu_{k}) \nonumber\\
	&=-L\tilde z_k-\tilde P^g_{k}+\sigma_{\mu}^{-1}(\tilde \mu_{k}- \mu_{k})\nonumber\\
	&\quad+L\cdot (\tilde z_{k}-  z_{k})+ \tilde P^g_{k}- P^g_{k}
\end{align}

Similarly, \eqref{SYNc2} is equal to 
\begin{align}
\label{SYNd2}
0 & = L \cdot \tilde\mu_{k}+ L \cdot (\tilde\mu_{k}-  \mu_{k}) + \sigma_{z}^{-1}(\tilde z_{k}- z_{k})
\end{align}

From the fact that $\mathcal{P}_{\Omega}(x)=({\rm{Id}}+N_{\Omega})^{-1}(x)$, \eqref{SYNc2} can be rewritten as $\tilde P^g_{k}=({\rm{Id}}+N_{\Omega})^{-1}$ $(P^g_{k} - \sigma_{g}(\nabla f(P^g_{k})+2\tilde \mu_{k} -\mu_{k}))$, or equivalently,
\begin{align}
\label{SYNd3}
	-\nabla f(P^g_{k})=2\tilde \mu_{k} -\mu_{k}+N_{\Omega}(\tilde P^g_{k})+\sigma_{g}^{-1}(\tilde P^g_{k}-P^g_{k})
\end{align}
Then, $\eqref{SYNd1}-\eqref{SYNd3}$ are rewritten as 
\begin{align}
\label{SYNc0}
	-\left[ {\begin{array}{*{20}{c}}
		L  \mu_{k}+P^d\\
		0\\
		\nabla f(P^g_{k})
		\end{array}} \right] = \left[ {\begin{array}{*{20}{c}}
		-\tilde P^g_{k}-L\tilde z_k\\
		L  \tilde\mu_{k}\\
		\tilde\mu_{k}+N_{\Omega}(\tilde P^g_{k})
		\end{array}} \right] +\Phi\left[ {\begin{array}{*{20}{c}}
		\tilde \mu_{k}- \mu_{k}\\
		\tilde z_{k}- z_{k}\\
		\tilde P^g_{k}-P^g_{k}
		\end{array}} \right]
\end{align}
where
\begin{align}\label{phi}
	\Phi=\left[ {\begin{array}{*{20}{c}}
		\sigma_{\mu}^{-1}I&L&I\\
		L& \sigma_{z}^{-1}I&0\\
		I&0&\sigma_{g}^{-1}I
		\end{array}} \right]
\end{align}

Define the following two operators
\begin{align}
	\label{OperU}
	&\mathcal{B}:\left[ {\begin{array}{*{20}{c}}
		\mu\\
		z\\
		P^g
		\end{array}} \right] \mapsto \left[ {\begin{array}{*{20}{c}}
		L  \mu +P^d\\
		0\\
		\nabla f(P^g)
		\end{array}} \right] \\
	\label{OperB}
	&\mathcal{U}:\left[ {\begin{array}{*{20}{c}}
		\mu\\
		z\\
		P^g
		\end{array}} \right] \mapsto \left[ {\begin{array}{*{20}{c}}
		- P^g-L z\\
		L  \mu\\
		\mu+N_{\Omega}( P^g)
		\end{array}} \right]
\end{align}

From \cite[Lemma 5.6]{yi2017distributed}, we know $({\rm{Id}}+\Phi^{-1}\mathcal{U})^{-1}$ exists and is single-valued. Denote $w^i=(\mu_i,z_i,P^{g}_i)$, $w=(w^i)$, $\tilde w=$ $(\tilde \mu, \tilde z_i,\tilde P^g)$. We show that $\eqref{SYNc1}-\eqref{SYNc4}$ can be  written as 
\begin{align}
	\label{Opera1}
	\tilde w_k&=\mathcal{T}(w_k)\\
	\label{Opera2}
	w_{k+1}&=w_k+\eta_k(\tilde w_k-w_k)
\end{align}
where the operator $\mathcal{T}=({\rm{Id}}+\Phi^{-1}\mathcal{U})^{-1}({\rm{Id}}-\Phi^{-1}\mathcal{B})$.

For \eqref{Opera1}, $\tilde w_k=\mathcal{T}(w_k)$ is equivalent to 
\begin{align}
	-\mathcal{B}(w_k)=\mathcal{U}(\tilde w_k)+\Phi\cdot(\tilde w_k-w_k)
\end{align}
This is exactly \eqref{SYNc0}. In addition, it is straightforward to see that $\eqref{SYNc4}-\eqref{SYNc6}$ are equivalent to \eqref{Opera2}. 

Equations  $\eqref{Opera1}-\eqref{Opera2}$ can be further rewritten as
\begin{align}
\label{Opera}
w_{k+1}&=w_k+\eta_k(\mathcal{T}(w_k)-w_k)
\end{align}

Denote $a_{min}=\min\{a_i\},\ a_{max}=\max\{a_i\}$, $\forall i 
\in \mathcal{N}$. 
%
We have the following result about the operator $\mathcal{T}$. Denote the maximal eigenvalues of $L$ by $\sigma_{\max}$. We have the following result.
\begin{lemma}\label{average_operator}
	Take $\zeta=\min\{\frac{1}{\sigma_{\max}^2}, \frac{a_{min}}{a_{max}^2}\}$, $\kappa>\frac{1}{2\zeta}$, and the step sizes  $\sigma_{\mu}, \sigma_{z}, \sigma_{g}$ such that $\Phi-\kappa I$ is positive semi-definite. Then we have the following assertions under $\Phi$-induced norm. 
	\begin{enumerate}
		\item $\mathcal T$ is an averaged operator, and $\mathcal{T}\in \mathcal{A}\left(\frac{2\kappa\zeta}{4\kappa \zeta-1}\right)$;
		\item there exists a nonexpansive operator $\mathcal{R}$ such that 
		\begin{align*}
		\mathcal{T}=\left(1-\frac{2\kappa\zeta}{4\kappa \zeta-1}\right){\rm{Id}}+\frac{2\kappa\zeta}{4\kappa \zeta-1}\mathcal{R}
		\end{align*}
		\item operators $\mathcal{T}$ and $\mathcal{R}$ have the same fixed points, i.e., $Fix (\mathcal{T})=Fix (\mathcal{R})$.
	\end{enumerate}
\end{lemma}
\begin{proof}
	For the assertion 1), we know $({\rm{Id}}+\Phi^{-1}\mathcal{U})^{-1}\in \mathcal{A}\left(\frac{1}{2}\right)$ and ${\rm{Id}}-\Phi^{-1}\mathcal{B} \in \mathcal{A}\left(\frac{1}{2\kappa\zeta}\right)$  from \cite[Lemma 5.6]{yi2017distributed}. Then, following from \cite[Proposition 2.4]{COMBETTES2015Compositions}, we know $\mathcal{T}\in \mathcal{A}\left(\frac{2\kappa\zeta}{4\kappa \zeta-1}\right)$. 
	
	From assertion 1) and definition of averaged operators, there exists a nonexpansive operator $\mathcal{R}$ such that 
	\begin{align}
	\label{nonexpansive}
	\mathcal{T}=\left(1-\frac{2\kappa\zeta}{4\kappa \zeta-1}\right){\rm{Id}}+\frac{2\kappa\zeta}{4\kappa \zeta-1}\mathcal{R}
	\end{align}
	Then, we have the assertion 2).
	
	Since $\mathcal{T}$ is $\frac{2\kappa\zeta}{4\kappa \zeta-1}$-averaged, $\mathcal{T}$ is also a nonexpansive operator \cite[Remark 4.24]{bauschke2011convex}. For any nonexpansive operator $\mathcal{T}$, $Fix(\mathcal{T})\neq\emptyset$ \cite[Theorem 4.19]{bauschke2011convex}.
	Suppose $x$ is a fixed point of $\mathcal{T}$, and we have $\mathcal{T}(x)=x=\left(1-\frac{2\kappa\zeta}{4\kappa \zeta-1}\right){\rm{Id}}(x)+\frac{2\kappa\zeta}{4\kappa \zeta-1}\mathcal{R}(x)$. Thus, $\frac{2\kappa\zeta}{4\kappa \zeta-1}{\rm{Id}}(x)=\frac{2\kappa\zeta}{4\kappa \zeta-1}\mathcal{R}(x)$, which is equivalent to $x=\mathcal{R}(x)$. 
	
	Similarly, suppose $x$ is a fixed point of $\mathcal{R}$, and we have 
	$\mathcal{T}(x)=\left(1-\frac{2\kappa\zeta}{4\kappa \zeta-1}\right){\rm{Id}}(x)+\frac{2\kappa\zeta}{4\kappa \zeta-1}\mathcal{R}(x)=x$. Thus, assertion 3) holds, which completes the proof.
\end{proof}

So far, we convert the synchronous algorithm into a fixed-point iteration problem with an averaged operator (see \eqref{Opera}). Moreover, we also construct a nonexpansive operator $\mathcal{R}$. it enables us to prove the convergence of the asynchronous algorithm ASDPD, as we explain in the next subsection.

\subsection{Optimality of the equilibrium point}
Considering dynamic system \eqref{ASYN}, we give the following  definition of its equilibrium point.
\begin{definition}
	\label{def:ep.1}
	A point $w^*=(w^*_i, i\in\mathcal{N})=(\mu^*_i, z^*_i, P^{g*}_i)$ 
	is an {equilibrium point} of  system \eqref{ASYN} if $\lim_{k_i\rightarrow +\infty} w_{k_i}=w_i^*$ holds for all $i$.
\end{definition}

Then, we have the following result.
\begin{theorem}
	\label{optimality}
	Suppose Assumption \ref{Slater}  holds. The component $P^{g*}, \mu^*$ of the equilibrium point $w^*$ is the primal-dual optimal solution to \eqref{eq:opt.general1}.
\end{theorem}

\begin{proof}
	By $\eqref{ASYN1}-\eqref{ASYN6}$ and Definition \ref{def:ep.1}, we have 
	\begin{subequations}
		\begin{align}
		\label{ASYNeq01}
		0&=- L \cdot \mu^* +L \cdot z^* +P^{g*}- P^d\\
		\label{ASYNeq02}
		0&= L \cdot \mu^* \\
		\label{ASYNeq03}
		-\nabla f(P^{g*})&=N_{\Omega}(P^{g*})+\mu^*	
		\end{align}
	\end{subequations}
	Then, we have 
	\begin{subequations}\label{equilibrium0}
		\begin{align}
		&0=\sum\limits_{i\in \mathcal{N}} P_i^{g* }- \sum\limits_{i\in \mathcal{N}}  P_i^d\\
		&\mu_{i}^*=\mu_{j}^*=\mu_0^* \\
		-\nabla f(P^{g*})&=N_{\Omega}(P^{g*})+\mu^*	
		\end{align}
		where $\mu_0^*$ is a constant. By \cite[Theorem 3.25]{ruszczynski2006nonlinear}, we know    \eqref{equilibrium0} is exactly the KKT condition of the problem \eqref{eq:opt.general1}.
	\end{subequations}
	In addition, \eqref{eq:opt.general1} is a convex optimization problem and Slater's condition holds, which completes the proof.
\end{proof}

\subsection{Convergence analysis of asynchronous algorithm}

In this subsection, we investigate the convergence of ASDPD. The basic idea is to treat ASDPD as a randomized block-coordinate fixed-point iteration problem with delayed information. And then the results in \cite{peng2016arock} can be applied. 

Define vectors $\phi_i\in \mathbb{R}^{3n}, i\in \mathcal{N}$. The $j$th entry of $\phi_i$ is  denoted by $[\phi_i]_j$. Define $[\phi_i]_j=1$ if the $j$th coordinate of $w$ is also a coordinate of $w^i$, and $[\phi_i]_j=0$, otherwise. Denote by $\varphi$  a random variable (vector) taking values in $\phi_i, i\in \mathcal{N}$. Then $\textbf{Prob}(\varphi=\phi_i)=1/n$ also follows a uniform distribution. Let $\varphi_k$  be the value of $\varphi$ at the $k$th iteration. Then, a randomized block-coordinate fixed-point
iteration for \eqref{Opera2} is given by
\begin{align}
	\label{random_fixed}
	w_{k+1}=w_k+\eta_k\varphi_k\circ(\mathcal{T}(w_k)-w_k)
\end{align}
where $\circ$ is the Hadamard product of two matrices. Here, we assume only one MG is activated at each iteration without loss of generality\footnote{Note that this model helps formulate the algorithm and analyze its convergence. In  implementation, we allow that two or more MGs are activated simultaneously, which can be modeled as two or more iterations in analysis.  }. 

Since  \eqref{random_fixed} is delay-free,  we further modify it for considering delayed information, which is 
\begin{align}
	\label{delay_random_fixed}
	w_{k+1}=w_k+\eta_k\varphi_k\circ(\mathcal{T}(\hat w_k)- w_k)
\end{align}
where $\hat w_k$ is the delayed information at iteration $k$. Note that, here, $k$ represents the \emph{global clock} defined in Section V. 
 We will show that Algorithm \ref{algorithm1} can be written as \eqref{delay_random_fixed} if $\hat w_k$ is properly defined.
Suppose MG $i$ is activated at the iteration $k$, then $\hat w_k$ is defined as follows. For MG $i$ and $j\in N_i$, replace $\mu_{i,k}$, $z_{i,k}$ and $P^g_{i,k}$ with $\mu_{i,k-\tau_{i}^k}$, $z_{i,k-\tau_{i}^k}$ and $P^g_{i,k-\tau_{i}^k}$. Similarly, replace $\mu_{j,k}$, $z_{j,k}$ with $\mu_{j,k-\tau_{j}^k}$ and $z_{j,k-\tau_{j}^k}$. With the random variable $\varphi$, variables of inactivated MGs are kept the same with the previous iterations. Then we have \eqref{delay_random_fixed}.

Next we make the following assumption.
\begin{assumption}
	\label{bounded_delay}
	The  time interval between two consecutive  iterations is bounded by $\chi$, i.e., $\sup_k\max_{i\in\mathcal{N}}\{\max\{\tau_i^k\}\}\le \chi$.
\end{assumption}

With the assumption, we have the convergence result.
\begin{theorem}\label{convergence}
	Suppose Assumptions \ref{Slater}, \ref{bounded_delay} hold. Take $\zeta=\min\{\frac{1}{\sigma_{\max}^2}, \frac{a_{min}}{a_{max}^2}\}$, $\kappa>\frac{1}{2\zeta}$, and the step-sizes $\sigma_{\mu}, \sigma_{z}, \sigma_{g}$ such that $\Phi-\kappa I$ is positive semi-definite. Choose $0<\eta_k< \frac{1}{1+2\chi/\sqrt{n}} \frac{4\kappa \zeta-1}{2\kappa\zeta}$. Then, with ASDPD, $P_k^g$ and $\mu_k$ converge to the primal-dual optimal solution of problem \eqref{eq:opt.general1} with probability 1.
\end{theorem}
\begin{proof}
	Combining \eqref{nonexpansive} and \eqref{delay_random_fixed}, we have 
	\begin{align}
	w_{k+1}&=w_k+\eta_k\varphi_k\circ\left(\left(1-\frac{2\kappa\zeta}{4\kappa \zeta-1}\right)\hat w_k \right. \nonumber\\
	&\qquad\qquad\qquad\qquad\left.- w_k+\frac{2\kappa\zeta}{4\kappa \zeta-1}\mathcal{R}(\hat w_k)\right) \nonumber\\
	\label{nonexpansive2}
	&=w_k+\eta_k\varphi_k\circ\left(\hat w_k- w_k	 \right.\nonumber\\
	&\qquad\qquad\qquad\qquad\left.+\frac{2\kappa\zeta}{4\kappa \zeta-1}(\mathcal{R}(\hat w_k)-\hat w_k)\right)
	\end{align}
	With $w_{i,k-\tau_{i}^k}=w_{i,k-\tau_{i}^k+1}=,\cdots,=w_{i,k}$, we have $\varphi_k\circ\left(\hat w_k- w_k\right)=0$. Thus, \eqref{nonexpansive2} is equivalent to 
	\begin{align}
	\label{nonexpansive3}
	w_{k+1}	&=w_k+\frac{2\eta_k\kappa\zeta}{4\kappa \zeta-1}\varphi_k\circ\left(\mathcal{R}(\hat w_k)-\hat w_k\right)
	\end{align}
	
	Invoking \cite{peng2016arock},  \eqref{nonexpansive3} with the nonexpansive operator $\mathcal{R}$ is essentially a kind of the  ARock algorithms suggested in \cite{peng2016arock}. Hence the convergence results given in that paper can directly be applied. Indeed,  Lemma 13 and Theorem 14 of \cite{peng2016arock} indicate that, the convergence of ARock is guaranteed by the condition 
	\begin{align}
	0<\frac{2\eta_k\kappa\zeta}{4\kappa \zeta-1}<\frac{1}{1+2\chi/\sqrt{n}}.
	\end{align} 
	Therefore, if $\eta_k$ satisfies $0<\eta_k< \frac{1}{1+2\chi/\sqrt{n}} \frac{4\kappa \zeta-1}{2\kappa\zeta}$, then $w_k$ converges to a random variable that takes value in the fixed points (denoted by $w_k^*$) of $\mathcal{R}$ with probability 1. Recalling $Fix (\mathcal{T})=Fix (\mathcal{R})$ and Theorem  \ref{optimality},  we know  $P_k^{g*}$ and $\mu_k^*$, as components of $w_k^*$, constitute the primal-dual optimal solution to the optimization problem \eqref{eq:opt.general1}. This completes the proof.
\end{proof}

Choose $\kappa=\frac{1}{2\zeta}+\epsilon$, where $\epsilon>0$ but very small. Then the upper bound of $\eta_k$ can be estimated by 
\begin{align*}
	\frac{1}{1+2\frac{\chi}{\sqrt{n}}} \frac{4\kappa \zeta-1}{2\kappa\zeta}&=\frac{1}{1+2\frac{\chi}{\sqrt{n}}} \frac{1+4\zeta \epsilon}{1+2\zeta\epsilon} \approx \frac{1}{1+2\frac{\chi}{\sqrt{n}}}
\end{align*}
Thus, there is $\eta_k<1$. Moreover, the upper bound of $\eta_k$ will decrease when the time delay increases, i.e., $\chi$ increases.

Given a fixed $\eta_k$ and a very small $\epsilon >0$, we have 
\begin{align}
	\label{max_Delay}
	\chi<\frac{\sqrt{n}(1-\eta_k)}{2\eta_k}	
\end{align}
Thus, the upper bound of acceptable time delay is approximately proportional to the square root of the number of MGs, which provides a helpful insight for controller design.

\section{Real-time Implementation}
The rapid variation of renewable generations and load demand requires that the controller can be implemented in real-time. In this section, the implementation of the ASDPD in both AC and DC MGs is introduced. First, we take the power system as a solver, which helps to eliminate the variables $ \tilde z$ and $z $ in ASDPD. Then, the real-time control diagram is illustrated. Finally, the optimality of results using such implementation method is proved.
\subsection{Taking the power system as a solver}
\subsubsection{Main idea}
Recalling the synchronous version \eqref{SYN} and \eqref{ASYNeq01}, the item $ \sum\limits_{j\in N_i}(z_{i,k}-z_{j,k}) $ in \eqref{SYN1} is utilized to balance the difference between $ P^g_{i,k}$ and  $P^d_i $. Denote $ \delta_{ij}=z_{j}-z_{i} $, and the last three terms of \eqref{SYN1} are $ P^g_{i,k}- P^d_i - \sum\limits_{j\in N_i}\delta_{ij,k} $.
From \eqref{ASYNeq01}, we know $0= \sum\limits_{j\in N_i}\delta_{ij,k}^* +P^{g*}_{i}- P^d_i $. This motivates us the power balance equation of each bus $ 0=P_i^{g*}-P_i^d -\sum\limits_{j\in N_i}P^*_{ij} $, where $ P^*_{ij} $ is the line power from bus $i$ to bus $j$ in the steady state. Thus, $ \delta_{ij} $ has the similar role to the line power $ P_{ij} $. This is a very important observation as $ P_{ij} $ is automatically implemented with physical dynamics of the power systems. We only need to measure it, which also implies that the computation of $\tilde z$, $z$ can be avoided. This is what we mean \textit{taking the power system as a solver}. Moreover, we can take $ P^g_{i,k}- P^d_i + \sum\limits_{j\in N_i}(z_{i,k}-z_{j,k}) $ as a whole and estimate it in both AC and DC MGs. 
\subsubsection{ AC MGs}
In AC MGs, the swing equation of AC bus $i$ is 
\begin{align}
	\label{omega_dynamic}
	M_i\dot \omega_i =P_i^g-P_i^d -D_i\omega_i -\sum\limits_{j\in N_i}P_{ij},\ i\in \mathcal {N}_{ac}
\end{align}
where, $M_i >0$, $D_i>0$ are constants, and $P_{ij}$ is the line power from bus $i$ to bus $j$. This model is suitable for both synchronous generators and converters \cite{de2016modular, Distributed_II:Wang, wang2018nonlinear}. \eqref{omega_dynamic} can be rewritten as 
\begin{align}
\label{omega_dynamic2}
 P_i^g- P_i^d-\sum\limits_{j\in N_i}P_{ij} =M_i\dot \omega_i+D_i\omega_i,\ i\in \mathcal {N}_{ac}
\end{align}
Thus, $ M_i\dot \omega_i+D_i\omega_i $ can be used to estimate $ P^g_{i,k}- P^d_i + \sum\limits_{j\in N_i}(z_{i,k}-z_{j,k}) $ in \eqref{SYN1} and its delayed form in \eqref{ASYN1}. 
By this control structure, the asynchronous algorithm \ref{algorithm1} is integrated to the real-time control in AC MGs.

\begin{figure}[t]
	\centering
	\includegraphics[width=0.3\textwidth]{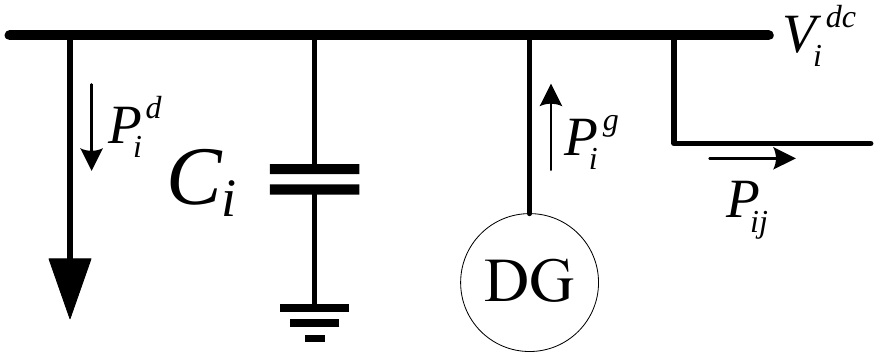}
	\caption{Simplified model of a DC MG}
	\label{DC_bus}
\end{figure}
\subsubsection{  DC MGs}
In DC MGs, DC capacitors are used to maintain the voltage stability of  DC buses \cite{wang2016distributed}. Then, the model of DC MGs can be simplified (see Fig.\ref{DC_bus}).
The power balance on DC bus $i$ is 
\begin{align}
	\label{V_dynamic}
	V_i^{dc} C_i\dot V_i^{dc} = P_i^g-P_i^d-\sum\limits_{j\in N_i}P_{ij},\ \ i\in \mathcal {N}_{dc}
\end{align}
where, $C_i$ is the capacitor connected to the DC bus; $V_i^{dc}$ is the voltage of the DC bus. 
Thus, $ V_i^{dc} C_i\dot V_i^{dc} $ can be used to estimate the $ P^g_{i,k}- P^d_i + \sum\limits_{j\in N_i}(z_{i,k}-z_{j,k}) $ in the DC MG. In this situation, we only need to measure the voltage, which is much easier to implement. Then, the asynchronous algorithm ASDPD is integrated to the real-time control in DC MGs. 

By taking the power system as a solver, the distributed asynchronous algorithm takes the following form
\begin{subequations}
	\label{RASYN}
	\begin{align}
	\label{RASYN1}
	&\tilde \mu_{i,k_i}= \mu_{i,k_i-\tau_{i}^{k_i}}+\sigma_{\mu}\bigg( - \sum\limits_{j\in N_i}\big(\mu_{i,k_i-\tau_{i}^{k_i}}-\mu_{j,k_j-\tau_{j}^{k_j}}\big) \nonumber \\
	&\qquad\qquad\qquad\qquad+ M_i\dot \omega_i+D_i\omega_i\big),\ i\in \mathcal {N}_{ac}
	\\
	\label{RASYN2}
	&\tilde \mu_{i,k_i}= \mu_{i,k_i-\tau_{i}^{k_i}}+\sigma_{\mu}\bigg( - \sum\limits_{j\in N_i}\big(\mu_{i,k_i-\tau_{i}^{k_i}}-\mu_{j,k_j-\tau_{j}^{k_j}}\big) \nonumber \\
	&\qquad\qquad\qquad\qquad+ V_i^{dc} C_i\dot V_i^{dc}\big),\quad i\in \mathcal {N}_{dc}
	\\
	\label{RASYN3}
	\begin{split}
	&\tilde P^g_{i,k_i}=\mathcal{P}_{\Omega_i}\left(P^g_{i,k_i-\tau_{i}^{k_i}} - \sigma_{g}\big(f_i^{'}(P^g_{i,k_i-\tau_{i}^{k_i}})+2\tilde \mu_{i,k_i} \right. \\ &\left.\qquad\qquad\qquad\qquad\qquad\qquad\quad \qquad-\mu_{i,k_i-\tau_{i}^{k_i}}\big)\right) \end{split}		\\
	\label{RASYN4}
	& \mu_{i,k_i+1}= \mu_{i,k_i-\tau_{i}^{k_i}} + \eta_k \big(\tilde \mu_{i,k_i}-\mu_{i,k_i-\tau_{i}^{k_i}}\big)\\
	\label{RASYN6}
	& P^g_{i,k_i+1}= P^g_{i,k_i-\tau_{i}^{k_i}}+ \eta_k \big(\tilde P^g_{i,k_i}-P^g_{i,k_i-\tau_{i}^{k_i}}\big)
	\end{align}
\end{subequations}
In the algorithm \eqref{RASYN}, only $ \mu $ needs to be transmitted between neighbors. Moreover, the variables $\tilde z$, $z$ are not necessary, which simplifies the controller greatly. Based on \eqref{RASYN}, we have the following \underline{r}eal-\underline{t}ime \underline{as}ynchronous \underline{d}istributed algorithm for \underline{p}ower \underline{d}ispatch (RTASDPD)
\begin{algorithm}
	\caption{\textit{RTASDPD}}
	\label{algorithmRT}
	\textbf{Input:} For MG $i$, the input is $\mu_{i,0} \in\mathbb{R}^n$, $P_{i,0}^g\in \Omega_i$.
	
	
	\textbf{Iteration at} \textit{$k_i$}: Suppose MG $i$'s clock ticks at time $k_i$, then MG $i$ is activated and updates its local variables as follows: 
	
	$\ \ $\textbf{Step 1: }\textbf{Reading phase}
	
	$\quad$Get $\mu_{j,k_j-\tau_{j}^{k_j}}$ from its neighbors' output cache. For an AC MG $ i $, measure the frequency $ \omega_i $. For a DC MG $ i $, measure the voltage $ V_i $. 
	
	$\ \ $\textbf{Step 2:} \textbf{Computing phase}
	
	$\quad$For $ i\in \mathcal {N}_{ac} $, calculate $\tilde \mu_{i,k_i}$ and $\tilde P^g_{i,k_i}$ according to \eqref{ASYN1} and \eqref{ASYN3} respectively. For $ i\in \mathcal {N}_{dc} $, calculate $\tilde \mu_{i,k_i}$ and $\tilde P^g_{i,k_i}$ according to \eqref{ASYN2} and \eqref{ASYN3} respectively.
	
	$\quad$Update $ \mu_{i,k_i+1}$ and $ P^g_{i,k_i+1}$ according to \eqref{ASYN4} and \eqref{ASYN5} respectively.
	
	$\ \ $\textbf{Step 3:} \textbf{Writing phase}
	
	$\quad$Write $\mu_{i,k_i+1}$ to its output cache and  $\mu_{i,k_i+1}$, $P^g_{i,k_i+1}$ to its local storage. Increase $k_i$ to $k_i+1$.
\end{algorithm}

\begin{remark}
	By taking the power system as a solver, there are three main advantages:
	\begin{itemize}
		\item Only the frequency in AC MG and voltage in DC MG need to be measured, which avoid the measurement of load demand $P_i^d$. As we know, the load demand is usually difficult to measure while the measurement of frequency and voltage is much easier. 
		\item We simplify the communication graph, where only the neighboring communication is needed. Moreover, we also simplify the controller structure. The auxiliary variables $\tilde z$ and $z$ are eliminated, making the controller easier to implement.
		\item In the problem \eqref{eq:opt.general1}, the power loss is not considered. In the real-time implementation, we measure the frequency and intend to drive it to the nominal value. In this way, the impact of power loss such as line and inverter loss can  be considered. 
	\end{itemize}
\end{remark}
\subsection{Control diagram}
\begin{figure}[t]
	\centering
	\includegraphics[width=0.49\textwidth]{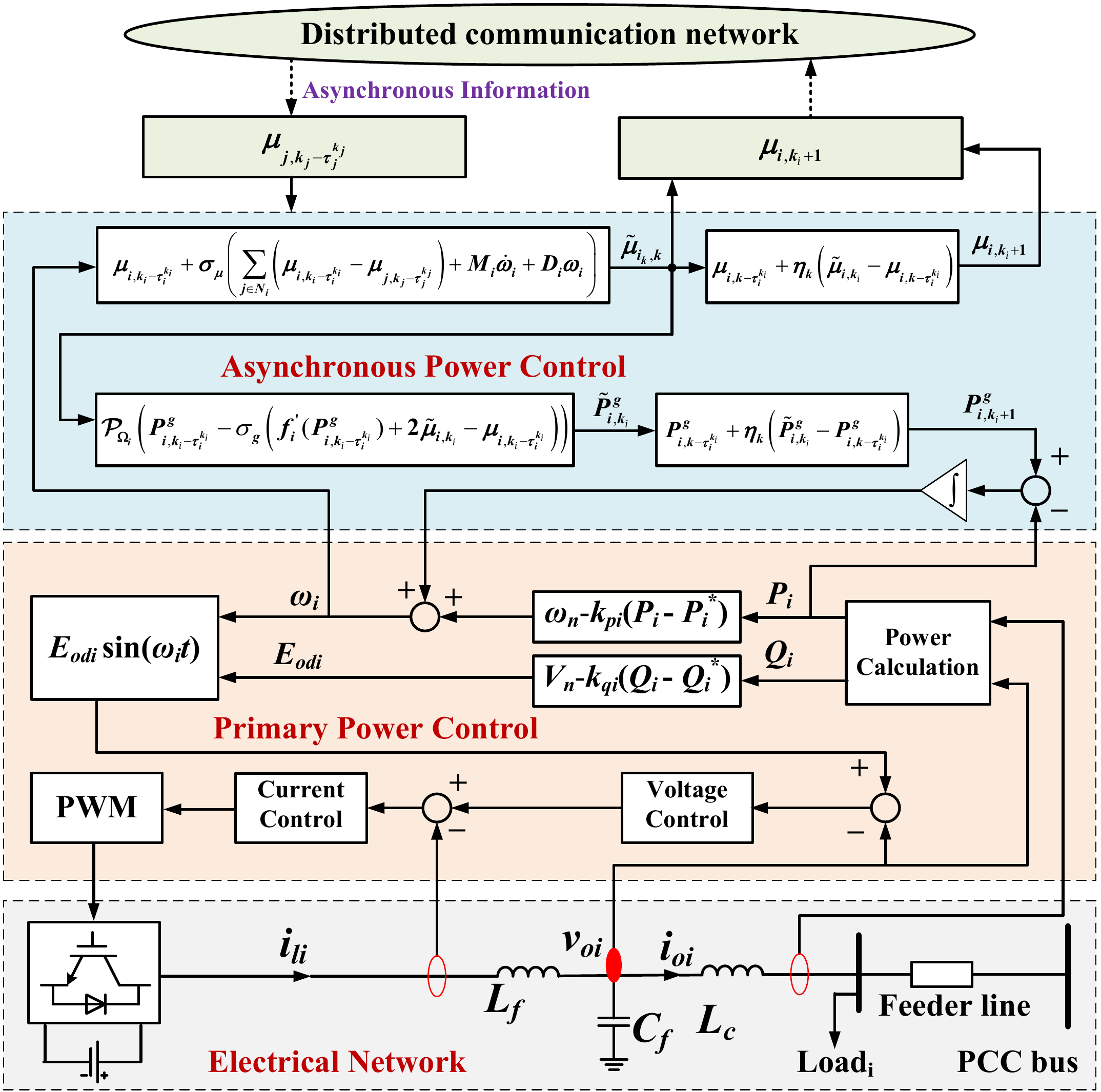}
	\caption{Control diagram of the proposed method}
	\label{Control_implementation}
\end{figure}
The control diagram is shown in Fig.\ref{Control_implementation}, which is composed of four levels: the electric network, the primary power control, the asynchronous power control and the distributed communication. In the electric network level, the current and voltage are measured as the input of the primary power control level. The primary power control level includes three loops, i.e., the current loop, the voltage loop and the power loop. In the power loop, droop control is utilized for both active power and reactive power control, where the active power and frequency are sent to the asynchronous power control level. Algorithm \ref{algorithm1} is integrated in the asynchronous control level, where the asynchronous information from neighboring  MGs is utilized. The output $P^g_{i,k+1}$ is the reference of the active power control. The error between $P^g_{i,k+1}$ and the measured active power is fed to the primary power control  via an integral operator. Other outputs $\mu_{i,k+1}$ is written to its output cache, which are sent to its neighbors via the  communication network. 

The  control diagram of DC MGs is similar to that in Fig.\ref{Control_implementation}, where the main difference is that the DC bus voltage $ V_i^{dc} $ is measured. The details are omitted here. 
\subsection{Optimality of the implementation}

By the implementation method, we claim that the equilibrium point of \eqref{RASYN} is also optimal with respect to the optimization problem \eqref{eq:opt.general1}. 
In steady state, we have $\omega_i=\omega_j=\omega^*, \forall i, j \in \mathcal{N}$, and $\frac{d V_i^{dc}}{d t}=0, \forall i\in \mathcal{N}_{dc}$. By $\eqref{ASYN1}-\eqref{ASYN4}$ and Definition \ref{def:ep.1}, we have 
	\begin{subequations}
		\begin{align}
		\label{ASYNeq1}
		\begin{split}
		&0=  \sum\limits_{j\in N_i}\left(\mu_{i}^*-\mu_{j}^*\right) + D_i\omega^*, \quad i\in\mathcal{N}_{ac}
		\end{split} \\
		\label{ASYNeq2}
		\begin{split}
		&0=  \sum\limits_{j\in N_i}\left(\mu_{i}^*-\mu_{j}^*\right) , \quad i\in\mathcal{N}_{dc}
		\end{split} \\
		\label{ASYNeq3}
		\begin{split}
		&P^{g*}_{i}=\mathcal{P}_{\Omega_i}\left(P^{g*}_{i} - \sigma_{g}\big(f_i^{'}(P^{g*}_{i})+\mu_{i}^*\big)\right) \end{split}	
		\end{align}
	\end{subequations}
	From \eqref{ASYNeq1} and \eqref{ASYNeq2}, we have 
	\begin{subequations}
		\label{equilibrium}
		\begin{align}
		r_1\mu^*+D_1\omega^*&=0\\
		&\vdots\nonumber\\
		r_{|\mathcal{N}_{ac}|}\mu^*+D_{|\mathcal{N}_{ac}|}\omega^*&=0\\
		r_{|\mathcal{N}_{ac}|+1}\mu^*&=0\\
		&\vdots\nonumber\\
		r_{|\mathcal{N}|}\mu^*&=0
		\end{align}
		where $r_i$ is the $i$th row of Laplacian matrix $L$, and $r_1+r_2+\cdots+r_{|\mathcal{N}|}=0$. Thus, we have
		\begin{align}
		\omega^*\sum\limits_{i\in\mathcal{N}_{ac}}D_i=0
		\end{align}
		This implies that $\omega^*=0$.
		Then, we have $\mu_{i}^*=\mu_{j}^*=\mu^*$ with a constant $\mu^*$. Other analysis is similar to that of Theorem \ref{optimality}, which is omitted here.
	\end{subequations}

\section{Case studies}
\begin{figure}[t]
	\centering
	\includegraphics[width=0.45\textwidth]{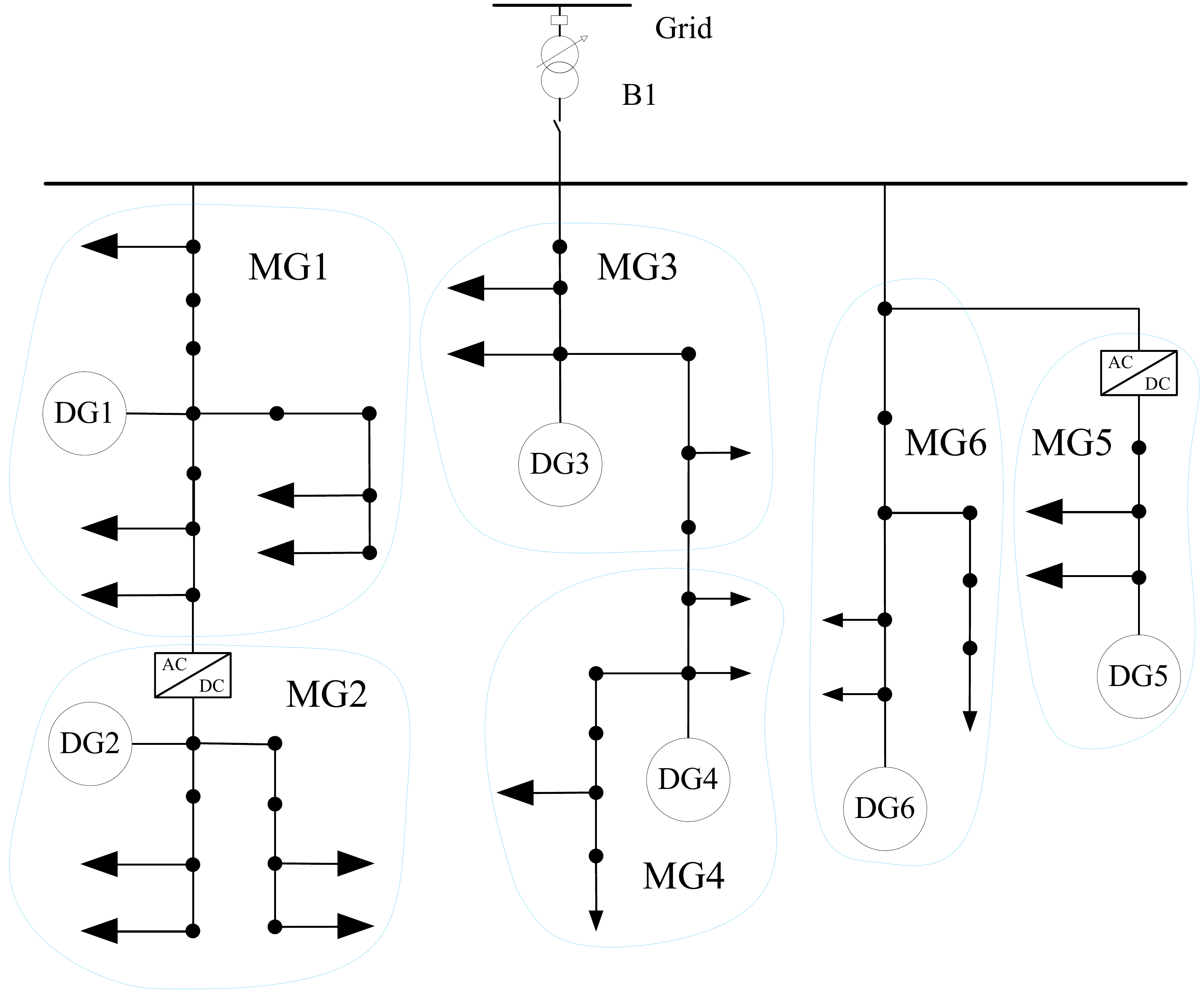}
	\caption{A schematic diagram of a typical 43-bus MG system}
	\label{fig:system}
\end{figure}
\subsection{System Configuration}
To verify the performance of the proposed method, a 44-bus system  shown in Fig.\ref{fig:system} is used for the test, which is a modified benchmark of low-voltage MG systems \cite{papathanassiou2005benchmark,wu2017distributed}. The system includes three feeders with six dispatchable MGs, where MG2 and MG5 are DC MGs while the others are AC MGs. The Breaker 1 is open, which implies that the system operates in an islanded mode. All simulations are implemented in the professional power system simulation software PSCAD.

The simulation scenario is: 1) at $t=2$s, there is a $60$kW load increase in the system; 2) at $t=8$s, there is a $30$kW load drop. Then, each MG increases its generation to balance the power and restore system frequency. Their initial generations are $(58.93, 46.94, 66.43, 59.95, 52.06, 55.09)$ kW. The communication graph is undirected, which is shown in Fig.\ref{fig:communication}. Other parameters are given in Table \ref{SysPara}.

\begin{table}[http]
	\footnotesize
	\centering
	\caption{System parameters}
	\label{SysPara}
	\begin{tabular}{c c c c c c c}
		\hline 
		DG $i$ 				& 1 	& 2 	& 3    & 4 	  & 5 	 &6\\
		\hline
		$a_i$ 				 	 & 0.8   & 1  	& 0.65   & 0.75  & 0.9  & 0.85 \\
		$b_i$ 					 & 0.01  & 0.01 & 0.014  & 0.012 & 0.01 & 0.01 \\
		$\overline P_i^g$ (kW)  & 85 	 & 80 	& 90     & 85    & 80   & 80\\
		$\underline P_i^g$ (kW)   & 0     & 0    & 0  	 & 0  	 & 0    & 0  \\
		\hline
	\end{tabular}
\end{table}

\begin{figure}[t]
	\centering
	\includegraphics[width=0.28\textwidth]{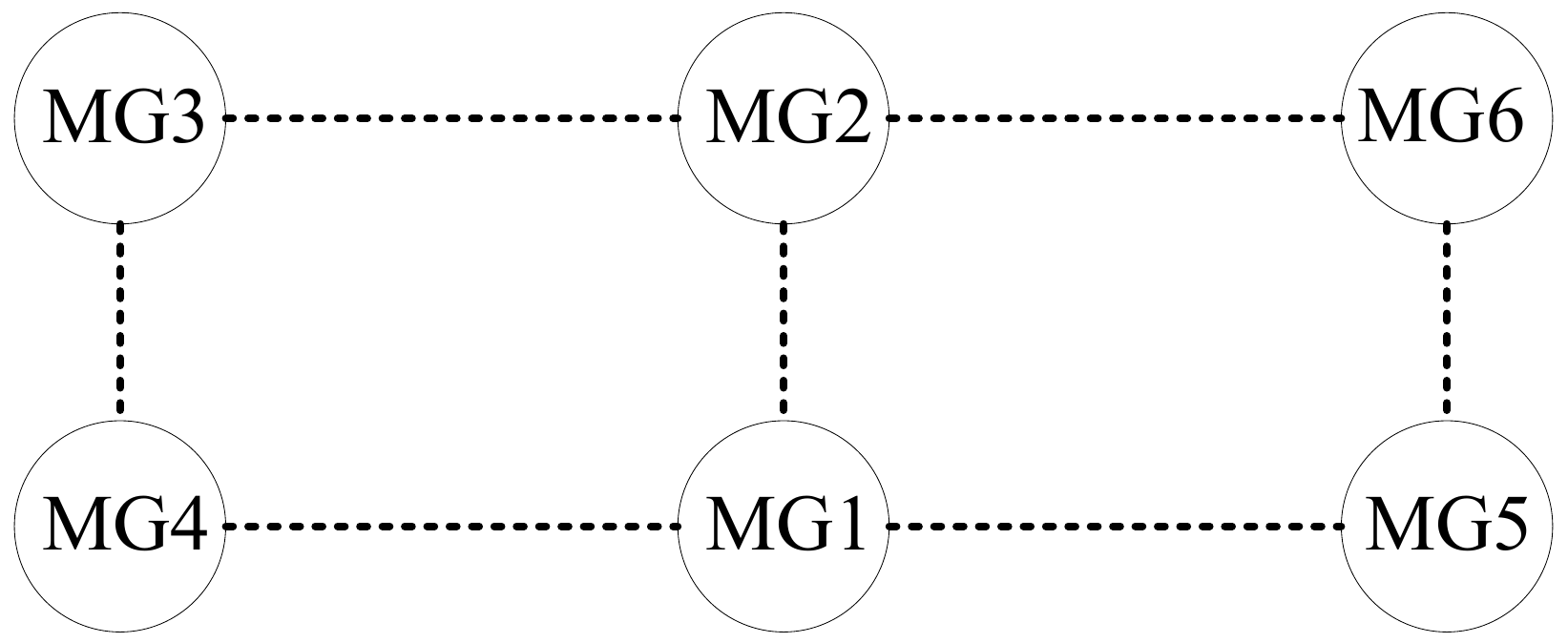}
	\caption{Communication graph of the system}
	\label{fig:communication}
\end{figure}

\subsection{Non-identical sampling rates}
Individual MGs may have different sampling rates (or control period) in practice, which could cause asynchrony and compromise the control performance. In this part, we consider the impact of non-identical sampling rates. The sampling rates of MG1-MG6 are set as 10,000Hz, 12,000Hz, 14,000Hz, 16,000Hz, 18,000Hz, 20,000Hz, respectively. The dynamics of the frequencies and voltages of MGs are shown in Fig.\ref{Figure_sampling_freq}.

\begin{figure}[t]
	\centering
	\includegraphics[width=0.49\textwidth]{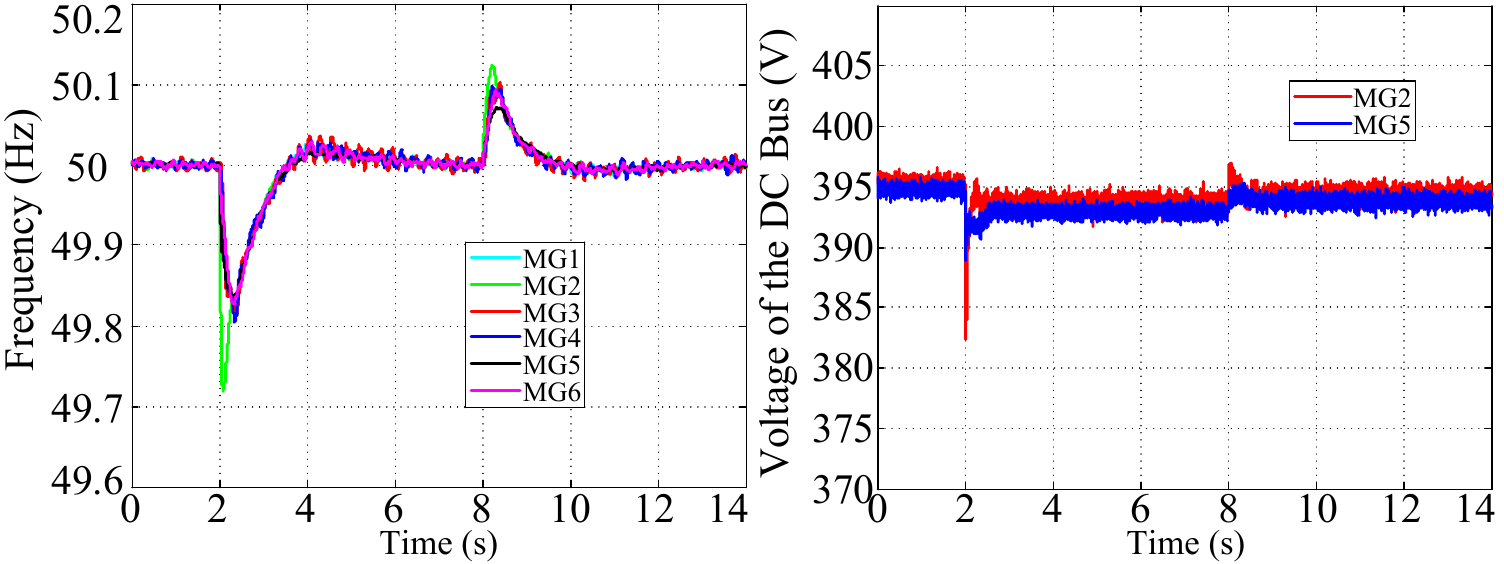}
	\caption{Dynamics of frequencies (left) and voltages (right). For a DC MG, its frequency implies the frequency of the corresponding DC/AC inverter.}
	\label{Figure_sampling_freq}
\end{figure}

As the load change is located in MG2, the frequency nadir of MG2 is the lowest (about 0.26Hz). The system frequency recovers in 4 seconds after the load change. When the load decreases, the frequency experiences an overshoot of 0.1Hz, and recovers in 2 seconds.  Voltages on the DC buses of MG2 and MG5 have a small drop when load increases. On the contrary, they voltages slightly increase after the load drop. The result demonstrates that the system  is fairly stable to load variation even with non-identical sampling rates.
\begin{figure}[t]
	\centering
	\includegraphics[width=0.49\textwidth]{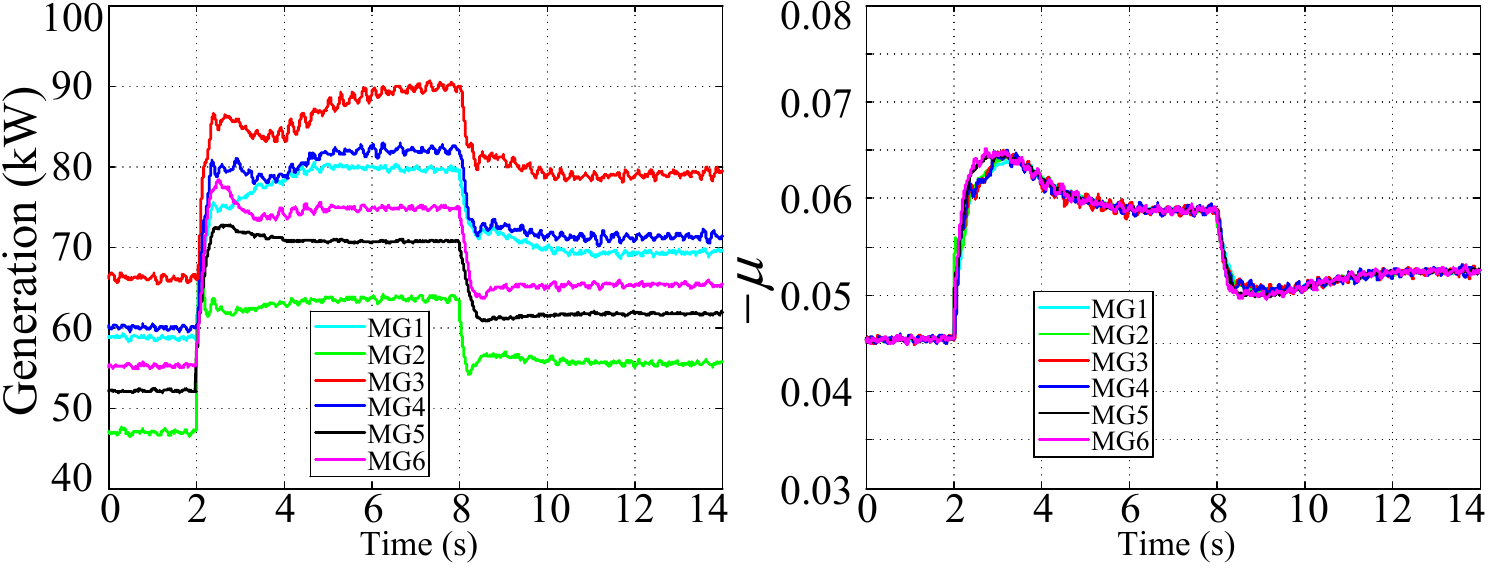}
	\caption{Dynamics of generations (left) and $-\mu$ (right).}
	\label{Figure_sampling_gen}
\end{figure}

Dynamics of generations and $-\mu$ are given in Fig.\ref{Figure_sampling_gen}. At the end of stage one (from 2s to 8s), generations of MGs are (79.32, 63.60, 90, 81.82, 70.47, 75.08)kW respectively. At the end of stage two (from 8s to 14s), their values are (69.50, 55.46, 79.20, 70.97, 61.86, 65.04)kW respectively. Generations are identical with that obtained by solving the centralized optimization problem (implemented by CVX). This result verifies the optimality of the proposed method. 
$-\mu_i$ stands for the marginal cost of MG $i$, whose dynamic is given on the right part of Fig.\ref{Figure_sampling_gen}. The marginal cost of different MGs converges to the same value when the system is stabilized, which indicates that the system operates in an optimal state. 

\subsection{Random time delays}
In practice, time delay always exists in the communication, which is usually varying up to channel situations. This implies that the time delay is random and cannot be known in advance. In this part, we examine the impact of time-varying time delays. Initially, all the time delays in communication are set as 20ms. And then we intentionally increase the time delays on the channels of MG1-MG2 and MG5-MG6. Additionally, we have the time delays on these two channels varying in ranges [100ms, 200ms], [200ms, 500ms], [500ms, 800ms] and [800ms, 1000ms], respectively, while the delays on other channels remain 20ms. Frequency and generation dynamics of MG1 under different scenarios are shown in Fig.\ref{Figure_sampling_delays}. It is observed that, When time delays increase, the convergence becomes slower with larger overshoots in frequency. However, the steady-state generations are still exactly identical to the optimal solution, which verifies the effectiveness of our controller under varying time delays. 

\begin{figure}[t]
	\centering
	\includegraphics[width=0.48\textwidth]{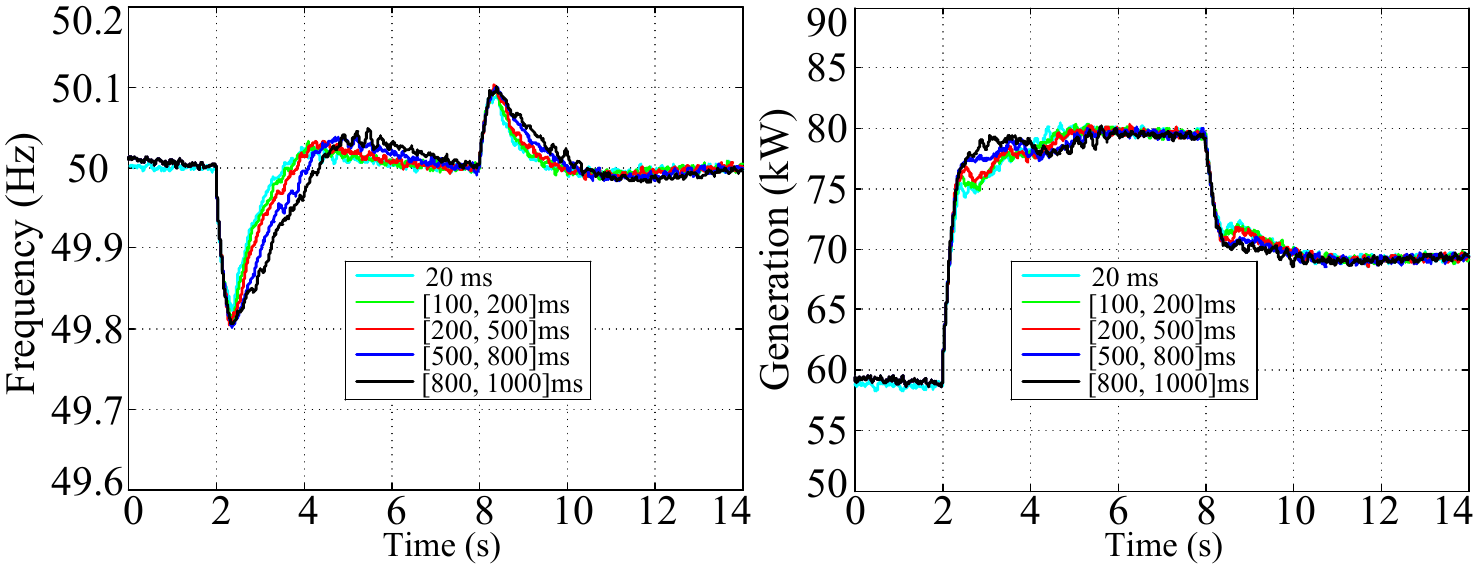}
	\caption{Frequencies and generations under different/varying time delays. }
	\label{Figure_sampling_delays}
\end{figure}

\subsection{Comparison with synchronous algorithm}

In this part, we compare the performances of the asynchronous and synchronous algorithms under imperfect communication. In the asynchronous case, the sampling rates of MGs are set to the same as that in Section VII.B and the time delay varies between $[500, 800]$ms. The  dynamics of MG1 with two algorithms are shown in Fig.\ref{Figure_samplings_async_sync}. With the synchronous algorithm SDPD, the system remains stable after load perturbations. However, the frequency nadir and overshoot deteriorate, and the convergence becomes slower. The generation  takes more time to reach the optimal solution, with an obvious fluctuation. The reason is that MGs have to wait for the slowest one in the synchronous case. This result confirms the advantage of  the asynchronous algorithm. 
\begin{figure}[t]
	\centering
	\includegraphics[width=0.48\textwidth]{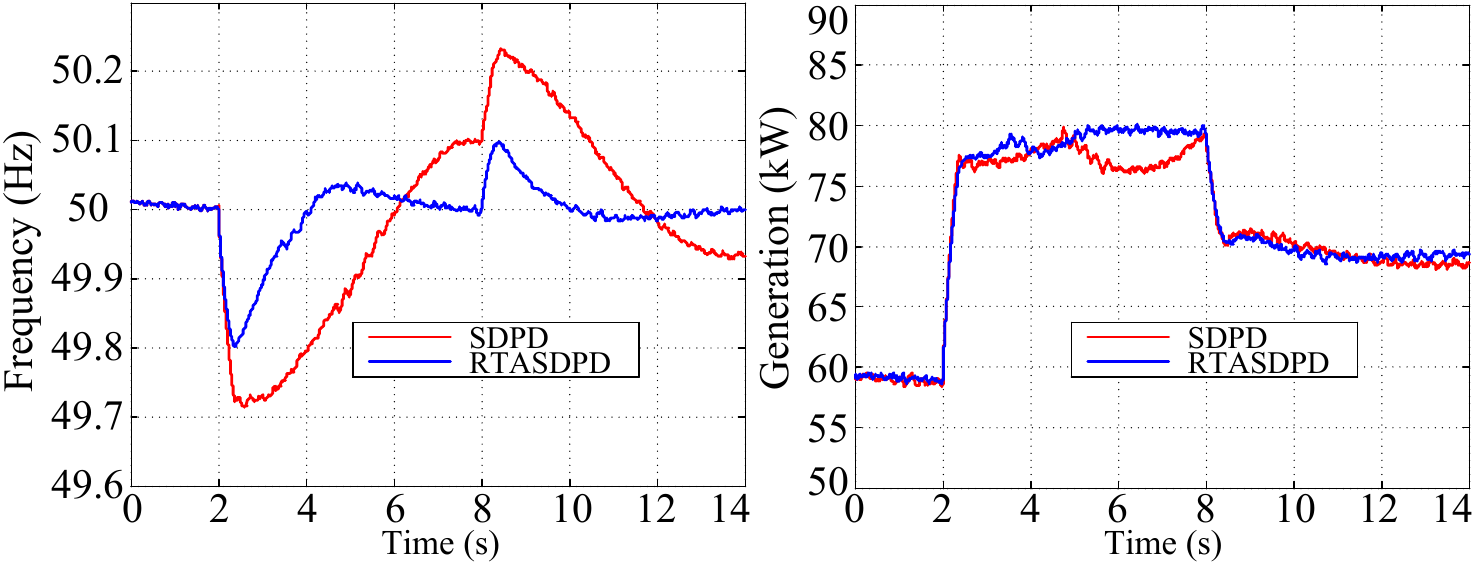}
	\caption{ Dynamics of frequencies and generations under synchronous and asynchronous cases.  }
	\label{Figure_samplings_async_sync}
\end{figure}


\subsection{Plug-n-play test}
In this part, we examine the  performance of RTASDPD under the plug-n-play operation mode. The simulation scenario is that DG4 is switched off at $t=2$s and switched on at $t=8$s. Dynamics of frequencies and generations are illustrated in Fig.\ref{Figure_pnp}. When DG4 is switched off, there is a small frequency oscillation. Since MG3 is close to DG4, its frequency nadir is the lowest. When DG4 is connected, the frequency oscillation is more fierce. However, the system is stabilized rapidly in 2s. Generation of DG4 drops to zero when it is switched off. Then other DGs increase their generations to re-balance the power. After DG4 is re-connected, all the generations recover to the initial values. This demonstrate that  our controller can adapt to the plug-n-play mode. 

\begin{figure}[t]
	\centering
	\includegraphics[width=0.48\textwidth]{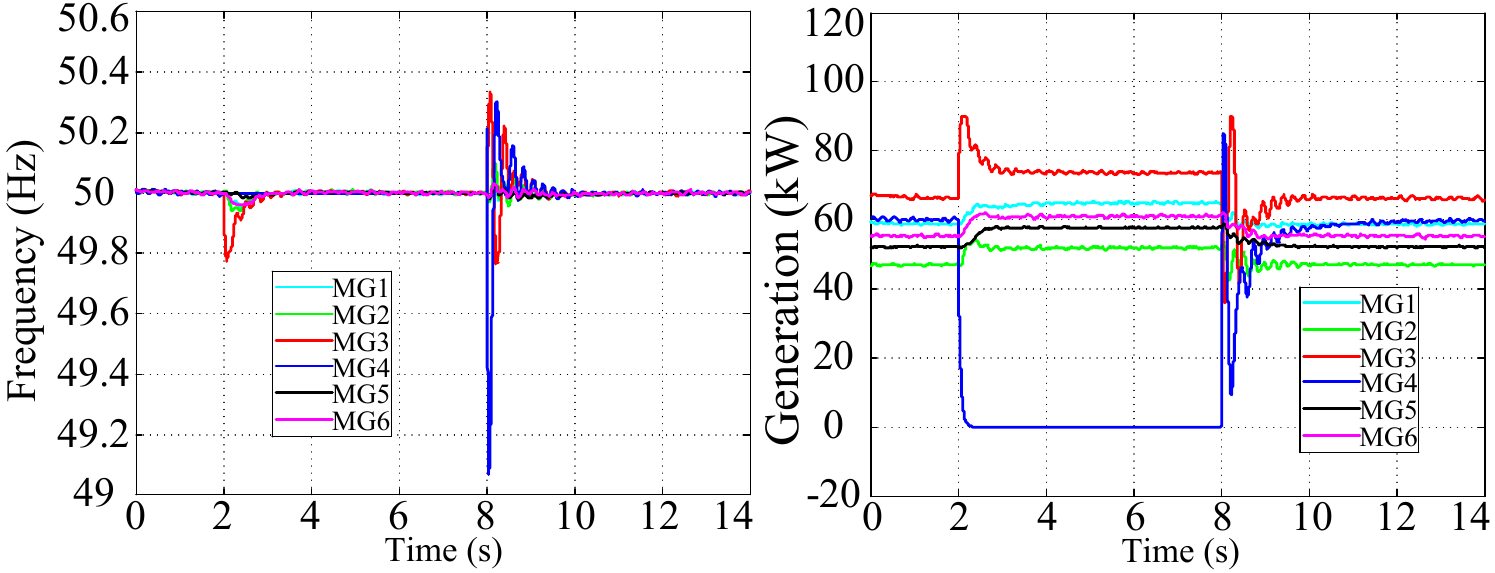}
	\caption{Dynamics of frequencies (left) and generations (right) when MG4 is switched off and on.}
	\label{Figure_pnp}
\end{figure}

\section{Conclusion}

In this paper, we have addressed the information asynchrony issue in the distributed optimal power control of hybrid MGs. By introducing a random clock, different kinds of asynchrony due to imperfect communication are fitted into a unified framework. Based on this, we have devised an asynchronous algorithm  with a proof of convergence. We have also provided an upper bound of the time delay. Furthermore, we have presented the real-time implementation of the asynchronous distributed power control in hybrid AC and DC MGs. In the implementation, the power system is taken as a solver, which simplifies the controller and can consider the power loss. Numerical experiments on PSCAD confirm the superior performance of the proposed methodology. 

Communication asynchrony widely exists in MGs. This paper gives a framework to design distributed controller under imperfect communication. The proposed methodology can also be extended to other related problems, such as voltage control in power systems and energy control in multi-energy systems. 

\bibliographystyle{IEEEtran}
\bibliography{mybib}

\ifCLASSOPTIONcaptionsoff
  \newpage
\fi

\end{document}